\documentclass[10pt]{amsart}
\UseRawInputEncoding
\usepackage[utf8]{inputenc}
\usepackage{amsfonts}
\usepackage{amssymb}
\usepackage{amsmath,amsthm}
\usepackage{amscd}
\usepackage{graphics}
\usepackage{cancel}
\usepackage{pdfsync}
\usepackage[british]{babel} 
\usepackage{mathrsfs}
\usepackage{enumitem}
\usepackage{stmaryrd}
\usepackage{mathrsfs}
\usepackage{wasysym}
\usepackage{latexsym}
\usepackage{xcolor}
\usepackage[colorlinks,linkcolor=blue,citecolor=blue,urlcolor=blue]{hyperref}

\numberwithin{equation}{section}

\newcommand{\beq}{\begin{equation}}
	\newcommand{\eeq}{\end{equation}}
\newcommand{\bea}{\begin{eqnarray}}
	\newcommand{\eea}{\end{eqnarray}}
\newcommand{\nn}{\nonumber}
\newcommand\noi{\noindent}
\newcommand{\tbf}{\textbf}

\newcommand{\rd}{\mathrm{d}}

\newcommand{\bk}{\begin{cases}}
	\newcommand{\ek}{\end{cases}}
\newcommand{\bs}{\boldsymbol}

\newtheorem{definition}{Definition}
\newtheorem{proposition}{Proposition}
\newtheorem{theorem}{Theorem}
\newtheorem{corollary}{Corollary}

\theoremstyle{definition}

\newtheorem{remark}{\textbf{Remark}}

\usepackage{enumerate}
\usepackage{float}
\usepackage{cancel}
\usepackage{hyperref}
\usepackage{mathtools}
\usepackage{afterpage}

\allowdisplaybreaks

\usepackage{bm}

\begin{document}
	
	\title[St\"ackel and Eisenhart lifts, Haantjes geometry and Gravitation]{St\"ackel and Eisenhart lifts, Haantjes geometry and Gravitation}
	
	\author{Ond\v{r}ej Kub\r{u}}
	\address{Instituto de Ciencias Matem\'aticas, C/ Nicol\'as Cabrera, No 13--15, 28049 Madrid, Spain}
	\email{
		ondrej.kubu@icmat.es}
	\author{Piergiulio Tempesta}
	\address{Departamento de F\'{\i}sica Te\'{o}rica, Facultad de Ciencias F\'{\i}sicas, Universidad
		Complutense de Madrid, 28040 -- Madrid, Spain \\ and Instituto de Ciencias Matem\'aticas, C/ Nicol\'as Cabrera, No 13--15, 28049 Madrid, Spain}
	\email{piergiulio.tempesta@icmat.es, ptempest@ucm.es}

	
	\date{August 6, 2026}
	
	\begin{abstract}
	We study lifts of integrable systems by means of generalized Stäckel geometry. To this end, we present the notion of Stäckel lift as a unified setting for the construction of new classes of integrable Hamiltonian systems of physical interest. The Stäckel lift extends the geometric framework underlying both the  Riemannian and the Lorentzian-type classical Eisenhart lifts.
		Moreover, we prove that Hamiltonian systems constructed through momentum-dependent Stäckel matrices are naturally endowed with a non-trivial symplectic-Haantjes structure.
		
		We further illustrate applications to magnetic systems separable in cylindrical coordinates; we describe them within  the Stäckel framework by means of modified Stäckel bases.

		Finally, we show that explicitly momentum-dependent lifting matrices generate Platonic-wave geometries with potential applications in modified gravity theories, or momentum-dependent metrics of generalized Hamilton geometries.

	\end{abstract}
	
	\maketitle

	\tableofcontents
	
	\section{Introduction}
	The program of geometrization of dynamics, inaugurated by Eisenhart in 1928 \cite{E1928AM}, has been one of the most fruitful research developments in modern theoretical and mathematical physics. 
	
	The central idea behind the Eisenhart lift (also known as the Eisenhart-Duval lift \cite{DBKP1985PRD,DGH1991PRD}) is to
	describe the dynamics of a mechanical system with configuration space coordinates $q^1,\ldots,q^n$, subject to a scalar potential $V(q^i,t)$ and a magnetic field with vector potential $\bs{A}(q^i,t)$, in terms of the geodesics of a Lorentzian metric defined on an $(n+2)$-dimensional space-time. We will refer to this construction as the Lorentzian-Eisenhart lift.

	This viewpoint can also be applied to the case of a lift from an $n$-dimensional to an $(n+1)$-dimensional space. In the literature, this is often referred to as the Riemannian-Eisenhart lift, since for a positive potential it allows one to define a Riemannian metric in the higher-dimensional space.
	
	\vspace{2mm}
	
	The purpose of this article is to construct new Hamiltonian models of physical interest by establishing a connection between  the theory of geometric lifts and generalized St\"ackel geometry. More precisely,  we extend the framework underlying the Eisenhart lift by means of a lifting procedure, rooted in St\"ackel geometry, which we consequently call the \textit{St\"ackel lift}.

From the seminal studies by St\"ackel \cite{St1893CRAS,St1897AMPA}, Di Pirro \cite{DPAMPA1896}, and Painlev\'e \cite{PainlCRAS1897} to the recent work by Magri \cite{MPhysD2023}, St\"ackel geometry has represented a fundamental theoretical setting for the formulation of the theory of separation of variables for classical Hamiltonian integrable systems.

	Given an $n$-dimensional manifold $Q$ and a  separable Hamiltonian system defined on $T^{*}Q$, the main idea of the approach we propose is to lift the $n\times n$ St\"ackel matrix associated with the system to a higher-dimensional cotangent bundle $T^{*}\tilde{Q}$.

	This is achieved in terms of an extended, generalized St\"ackel matrix, the $(n+1)\times (n+1)$ \textit{lifting matrix}. It includes the original St\"ackel matrix and contains  in an additional row and column suitably chosen functions depending on a new pair of conjugate variables $(q^{n+1}, p_{n+1})$  with respect to the extended symplectic form $\tilde{\omega}= \sum_{i=1}^{n+1} dp_i \wedge dq^{i}$. Consequently, $(T^{*}\tilde{Q}, \tilde{\omega})$ acquires the structure of a $2(n+1)$-dimensional symplectic manifold.
	
	In this way, we define a \textit{generalized St\"ackel matrix}, which explicitly depends on the new momentum $p_{n+1}$. This procedure can be naturally iterated by introducing new pairs of conjugate variables and higher-dimensional lifting matrices, with the original Stäckel matrix appearing as a block in each iteration.
	
	This procedure allows us to generate infinitely many higher-dimensional integrable models, starting from a \textit{seed system}. In fact, the St\"ackel lifting procedure provides at once and in a straightforward way all the Hamiltonian functions of the lifted Hamiltonian system. Under a suitable choice of the lifting matrix, these models can be interpreted as geodesic Hamiltonians defined in a suitable pseudo-Riemannian space. However, it is also possible to generate families of non-geodesic integrable Hamiltonian systems possessing higher-order integrals of motion, which could give rise to nontrivial higher-rank Killing tensors.
	
	In the interesting work \cite{BSPLA2011}, integrable systems with momentum-dependent separation relations have been studied on a manifold with fixed dimension $n$. Our work in a sense parallels this research by defining geometric lifts that systematically generate higher-dimensional systems with momentum-dependent Stäckel matrices, naturally connecting  to Haantjes geometry and physical applications including gravitational waves and magnetic systems.	In fact, our goal is to show that the St\"ackel lift represents a flexible tool that can be used to generate several typologies of geometric lifts. In Theorem \ref{theo:1}, we will prove that, under mild hypotheses, the Eisenhart lift can be regarded as an instance of the St\"ackel lift.
	
	As a second objective of this article, we shall clarify the relation among Eisenhart and St\"ackel lifts and Haantjes geometry. 
	
	The study of Haantjes tensors has experienced a resurgence of interest in recent years, initiated by the pioneering work of Ferapontov and Marshall \cite{FM2007MA} on the Haantjes approach to integrable hydrodynamic systems. Recent developments include the different manifold structures proposed in \cite{AT2025prepr,BKM2025,CMP2025prepr,MGall13,MGall15,
		MMcSV2024}. In particular, Haantjes tensors have been recently related to symplectic geometry in the theory of $\omega \mathscr{H}$ manifolds \cite{TT2021JGP,TT2022AMPA,RTT2022CNS,KRTT2024PRA,RTT2024AMPA}. This geometric framework offers a natural language for the formulation of the theory of Hamiltonian integrable systems, and the determination of their total and partial separation variables.

	As one of our main results, in Theorem \ref{maintheorem} we prove that the St\"ackel (and, in particular, the Eisenhart) lift of a given Hamiltonian integrable system endows $T^{*}\tilde{Q}$ with the structure of a nontrivial symplectic-Haantjes manifold. Coherently with the geometric idea of the St\"ackel lift, the symplectic-Haantjes structure can be reconstructed from the one associated with the original system. Interestingly enough, in the case of a St\"ackel lift explicitly dependent on the momenta, the Haantjes tensors associated with the lifted Hamiltonian model are not projectable along the fibers of $T^{*}\tilde{Q}$.

The lifts we  shall discuss find obvious applications in classical mechanics because they inherently produce Hamiltonian systems that are, by construction, both \textit{integrable} (in a dense open subset of the extended phase space) and \textit{separable}.
The abundance of these systems obviously reflects the arbitrariness inherent in the procedure of lifting a given model to a higher-dimensional system.

The present approach can be naturally related to the Jacobi-Sklyanin separation equations \cite{Skl}. In fact, as we will show, these equations can be readily obtained \textit{a posteriori}, once again as a direct consequence of the generalized Stäckel geometry.

	However, this observation does not exhaust their relevance. In fact, we will show that, in the case of a momentum-dependent lifting matrix, we can produce nontrivial Hamiltonians having a straightforward interpretation in gravitational contexts. More precisely, we obtain Hamiltonian systems representing Platonic waves \cite{BM2013PRD}, and in particular, \textit{gravitational pp-waves}.
	
	We recall that a plane-fronted wave with parallel rays or $pp$-wave is a Lorentzian spacetime geometry admitting a null, parallel vector field with flat wave fronts. The broader class of Platonic waves is conformal to Bargmann-Eisenhart waves, whose  wave fronts may not be flat as for $pp$-waves. Under certain additional conditions, Bargmann-Eisenhart waves are exact solutions of Einstein's equations or of $f(R)$ gravity theories, which extend Einstein's theory through higher-order curvature terms.
	
	Moreover, we can naturally define nontrivial geometric structures known as \textit{Hamilton} and \textit{generalized Hamilton spaces} \cite{Miron2001,RSS2024}. These spaces are endowed with a momentum-dependent metric. Their corresponding analogues in the configuration space are known as Lagrange and generalized Lagrange spaces \cite{Voicu2021}. Owing to their connection with the Zermelo wind problem \cite{Zermelo} and its generalizations \cite{Aldea}, these geometries find applications in acoustics, optics, quantum control, and quantum mechanics (see e.g. \cite{Li} for references). Recently, Hamilton spaces have been investigated as effective structures for configuration and phase spaces in quantum gravity phenomenological models involving Lorentz invariance violation or noncommutative spacetime, in particular deformed special relativity \cite{BBGLP2015PRD}, and in connection with prospective mesoscopic-scale signals of quantum gravity \cite{CLR2022PRD} (see also the review \cite{Albuquerque}).

	The structure of the article is as follows. In section~\ref{sec:2}, we review the most relevant aspects of the theory of $\omega\mathscr{H}$ manifolds that we need in the forthcoming discussion. In section~\ref{sec:3}, we discuss the notion of St\"ackel lift and, in Theorem \ref{theo:1}, we clarify its relation with the Eisenhart lift. In section~\ref{sec:4}, we elucidate the relationship between St\"ackel lifts and Haantjes geometry, and prove that momentum-dependent Stäckel lifts generate non-projectable Haantjes structures (Theorem \ref{maintheorem}). In section~\ref{sec:5}, we study an iterative procedure allowing us to construct higher-dimensional Hamiltonian models.  Several applications to physically relevant geometries, including gravitational waves, are discussed in section~\ref{sec:6}.  The study of non-geodesic Hamiltonians arising from the St\"ackel lift is performed in section~\ref{sec:7}. An extension of the theory to the case of magnetic systems separable in cylindrical coordinates and a discussion of the specific basis transformations needed are proposed in section~\ref{sec:8}. Future research perspectives are discussed in the final section~\ref{sec:9}.

	\section{Preliminaries on the Haantjes geometry}
	\label{sec:2}

	 In this section, we review basic notions concerning the geometry of Nijenhuis or Haantjes torsions, grounded in the foundational works \cite{Haa1955,Nij1951,Nij1955,Nij19552,FN1956}, and their connection to symplectic geometry \cite{TT2021JGP,TT2022AMPA,RTT2022CNS,RTT2024AMPA}. We shall focus only on the aspects of the theory relevant for the subsequent discussion.
	
	Let $M$ be a smooth manifold of dimension $n$, $\mathfrak{X}(M)$ be the Lie algebra of vector fields on $M$, and let $\boldsymbol{K}:\mathfrak{X}(M)\rightarrow \mathfrak{X}(M)$ be a smooth $(1,1)$-tensor field. In the following, all tensors will be considered to be smooth.
	\begin{definition} \label{def:N}
		The
		\textit{Nijenhuis torsion} of $\boldsymbol{K}$ is the vector-valued $2$-form
		\begin{equation} \label{eq:Ntorsion}
			\tau_ {\boldsymbol{K}} (X,Y):=\boldsymbol{K}^2[X,Y] +[\boldsymbol{K}X,\boldsymbol{K}Y]-\boldsymbol{K}\Big([X,\boldsymbol{K}Y]+[\boldsymbol{K}X,Y]\Big),
		\end{equation}
		where $X,Y \in \mathfrak{X}(M)$ and $[ \, \ ]$ denotes the commutator of two vector fields.
	\end{definition}
	\begin{definition} \label{def:H}
		\noi The \textit{Haantjes torsion} of $\boldsymbol{K}$ is the vector-valued $2$-form
		\begin{equation} \label{eq:Haan}
			\mathcal{H}_{\boldsymbol{K}}(X,Y):=\boldsymbol{K}^2 \tau_{\boldsymbol{K}}(X,Y)+\tau_{\boldsymbol{K}}(\boldsymbol{K}X,\boldsymbol{K}Y)-\boldsymbol{K}\Big(\tau_{\boldsymbol{K}}(X,\boldsymbol{K}Y)+\tau_{\boldsymbol{K}}(\boldsymbol{K}X,Y)\Big).
		\end{equation}
	\end{definition}
	
	The Haantjes (Nijenhuis) vanishing condition inspires the following definition.
	\begin{definition}
		A Haantjes (Nijenhuis) operator is a (1,1)-tensor field whose Haantjes (Nijenhuis) torsion identically vanishes.
	\end{definition}

	A simple, relevant example of Haantjes operator is that of a tensor field $\boldsymbol{K}$ which takes a diagonal form in a local chart $\boldsymbol{x}=(x^1,\ldots,x^n)$:
	\begin{equation}
		\boldsymbol{K}(\boldsymbol{x})=\sum _{i=1}^n \lambda_{i }(\boldsymbol{x}) \frac{\partial}{\partial x^i}\otimes \rd x^i, \label{eq:Ldiagonal}
	\end{equation}
	where $\lambda_{i }(\boldsymbol{x}):=\lambda^{i}_{i}(\boldsymbol{x})$ are the eigenvalue fields of $\boldsymbol{K}$ and $\left(\frac{\partial}{\partial x^1},\ldots, \frac{\partial}{\partial x^n}\right) $ are the fields forming the so called \textit{natural frame} associated with the local chart $(x^{1},\ldots,x^{n})$. As is well known, the Haantjes torsion of the diagonal operator \eqref{eq:Ldiagonal} vanishes. 
	
	We also recall that two frames $\{X_1,\ldots,X_n\}$ and $\{Y_1,\ldots,Y_n\}$ are said to be equivalent if $n$ nowhere vanishing smooth
	functions $f_i$ exist, such that
	\begin{equation}
		X_i= f_i(\boldsymbol{x}) Y_i \, , \qquad\qquad i=1,\ldots,n \ .
	\end{equation}
	
	\begin{definition}\label{def:Iframe}
		An \emph{integrable} frame is a reference frame equivalent to a natural frame.
	\end{definition}

	\par
	\noi Many more examples of Haantjes operators, relevant in classical mechanics and in Riemannian geometry, can be found, for instance, in the works \cite{KRTT2024PRA}, \cite{RTT2022CNS}--\cite{RTT2024AMPA} and \cite{TT2021JGP}--\cite{TT2016SIGMA}.
	
	\subsection{Haantjes algebras}

	Another crucial notion is that of Haantjes algebras, which have been introduced in \cite{TT2021JGP} and further studied and generalized in \cite{RTT2023JNS}.

	\begin{definition}\label{def:HM}
		A Haantjes algebra is a pair $(M, \mathscr{H})$ satisfying the following properties:
		\begin{itemize}
			\item
			$M$ is a differentiable manifold of dimension $\mathrm{n}$;
			\item
			$ \mathscr{H}$ is a set of Haantjes operators $\boldsymbol{K}:\mathfrak{X}(M)\to \mathfrak{X}(M)$. Also, they generate:
			\begin{itemize}
				\item
				a free module over the ring of smooth functions on $M$:
				\begin{equation}
					\label{eq:Hmod}
					\mathcal{H}_{\left( f\boldsymbol{K}_{1} +
						g\boldsymbol{K}_2 \right)}(X,Y)= \mathbf{0}
					\, \qquad\forall\, X, Y \in \mathfrak{X}(M) \, \quad \forall\, f,g \in C^\infty(M)\,\quad \forall ~\boldsymbol{K}_1,\boldsymbol{K}_2 \in \mathscr{H};
				\end{equation}
				\item
				a ring with respect to the composition operation
				\begin{equation}
					\label{eq:Hring}
					\mathcal{H}_{(\boldsymbol{K}_1 \, \boldsymbol{K}_2)}(X,Y)=\mathbf{0} \, \qquad
					\forall\, \boldsymbol{K}_1,\boldsymbol{K}_2\in \mathscr{H}, \quad\forall\, X, Y \in \mathfrak{X}(M)\ .
				\end{equation}
			\end{itemize}
		\end{itemize}
		Moreover, if 
		\begin{equation}
			\boldsymbol{K}_1\,\boldsymbol{K}_2=\boldsymbol{K}_2\,\boldsymbol{K}_1 \, \quad\qquad\ \forall~\boldsymbol{K}_1,\boldsymbol{K}_2 \in \mathscr{H}\,
		\end{equation}
		the algebra $(M, \mathscr{H})$ is said to be an Abelian Haantjes algebra. 
	\end{definition}
	
	In essence, we can think of $\mathscr{H}$ as an associative algebra of Haantjes operators.
	
	\begin{remark}
		A fundamental fact concerning Abelian Haantjes algebras is that there exist sets of coordinates, called \textit{Haantjes coordinates}, which allow us to write simultaneously all $\boldsymbol{K}\in \mathscr{H}$ in a block-diagonal form. In particular, if $\mathscr{H}$ is an algebra of semisimple operators, then in Haantjes coordinates all $\boldsymbol{K}\in \mathscr{H}$ acquire a purely diagonal form \cite{TT2021JGP}. 
		
	\end{remark}
	
	\vspace{2mm}

	\subsection{The symplectic-Haantjes manifolds}

	Symplectic-Haantjes (or $\omega \mathscr{H}$) manifolds are symplectic manifolds endowed with a compatible algebra of Haantjes operators \cite{TT2021JGP,TT2022AMPA,TT2022CMP}. Apart from their mathematical properties, these manifolds are relevant because they provide a natural setting in which the theory of Hamiltonian integrable systems can be naturally formulated \cite{KRTT2024PRA,RTT2022CNS,RTT2024AMPA}.

	\begin{definition}\label{def:oHman}
		A symplectic--Haantjes (or $\omega \mathscr{H}$) manifold of class $m$ is a triple $( M,\omega,\mathscr{H})$ which satisfies the following properties:
		\begin{itemize}
			\item[(i)]
			$(M,\omega)$ is a symplectic manifold of dimension $ 2 \, n$;
			\item[(ii)]
			$\mathscr{H}$ is an Abelian Haantjes algebra of rank $m$;
			\item[(iii)]
			$(\omega,\mathscr{H})$ are algebraically compatible, that is
			\begin{equation}
				\omega(X,\boldsymbol{K} Y)=\omega(\boldsymbol{K} X,Y) \qquad \forall \boldsymbol{K} \in \mathscr{H}\ .
			\end{equation}
		\end{itemize}
	\end{definition}
	
We recall the following
	\begin{proposition}\label{lm:aupari} \cite{TT2022AMPA}
		Given a $2n$-dimensional $\omega \mathscr{H}$ manifold $M$, every generalized eigen--distribution $\ker(\boldsymbol{K}-\lambda_{i } \boldsymbol{I})^{{\rho_i}}$, $\rho_i \in \mathbb{N}\backslash\{0\}$, is of even rank.
		Therefore, the geometric and algebraic multiplicities of each eigenvalue $\lambda_{i }(\boldsymbol{x})$ are even.
	\end{proposition}

	Let us consider the spectral decomposition of the tangent spaces of $M$ given by
	\begin{equation}\label{def:Di}
		T_{\boldsymbol{x}}M=\bigoplus_{i=1}^s \mathcal{D}_i(\boldsymbol{x}), \qquad \qquad
		\mathcal{D}_i(\boldsymbol{x}) :=
		\ker\big(\boldsymbol{K}-\lambda_i\boldsymbol{I}\big)^{\rho_i}(\boldsymbol{x}),
	\end{equation}
	where $\rho_i$ is the Riesz index of the eigenvalue $\lambda_i$. By the previous Proposition, since $rank$ $\mathcal{D}_i\geq 2, \ i=1,\ldots s$, we conclude that the number of distinct eigenvalues of any
	Haantjes operator $\boldsymbol{K}$ of an $\omega\mathscr{H}$ manifold is not greater than $n$.

	\subsection{Haantjes chains}
	\noi We now recall the notion of Haantjes chains \cite{TT2022AMPA, TTPRE2012}, which will play a central role in our subsequent analysis.
	
	\begin{definition}
		Let $( M,\mathscr{H})$ be a Haantjes algebra of rank $m$. We shall say that a smooth function $H$ generates a Haantjes chain of closed 1-forms of length $m$ if there exists
		a distinguished basis $\{\tilde{\boldsymbol{K}}_1,\ldots, \tilde{\boldsymbol{K}}_m\}$ of $\mathscr{H}$
		such that
		\begin{equation} \label{eq:MHchain}
			\rd (\tilde{\boldsymbol{K}}^T_\alpha \,\rd H )=\boldsymbol{0} \, \quad\qquad \alpha=1,\ldots,m \,
		\end{equation}
		where $\tilde{\boldsymbol{K}}^{T}_{\alpha}: \Omega^{1}(M) \to \Omega^{1}(M)$ is the transposed operator of $\tilde{\boldsymbol{K}}_{\alpha}$. The (locally) exact 1-forms 
		\begin{equation}
			\rd H_\alpha=\tilde{\boldsymbol{K}}^T_\alpha \,\rd H,
		\end{equation}
		which are supposed to be linearly independent, are said to be the elements of the Haantjes chain of length $\mathrm{m}$ generated by $H$. The functions $H_\alpha$ are the potential functions of the chain.
	\end{definition}

	In order to enquire about the existence of Haantjes chains for an assigned Haantjes algebra, we have to consider the co-distribution, of rank $r\leq m$, generated by a given function $H$ via an arbitrary basis
	$ \{ \boldsymbol{K}_1, \boldsymbol{K}_2,\ldots,\boldsymbol{K}_{m}\} $ of $\mathscr{H}$, i.e.
	\begin{equation} \label{eq:codKH}
		\mathcal{D}_H^\circ:=Span\{ \boldsymbol{K}_1^T \rd H, \boldsymbol{K}_2^T \rd H,\ldots,\boldsymbol{K}_{m}^T\,\rd H\} \,
	\end{equation}
	and the distribution $\mathcal{D}_H$ of the vector fields annihilated by them, which has rank $(n-r)$.
	Note that such distributions do not depend on the particular choice of the basis of $\mathscr{H}$.
	
	The following result clarifies the compatibility between a given $\omega \mathscr{H}$ manifold and a set of functions in involution.
	\begin{corollary}
		A function $H$ generates a Haantjes chain only if it satisfies the invariance condition
		\begin{equation}
			\boldsymbol{K}^T (\mathcal{D}_H^\circ ) \subseteq \mathcal{D}_H^\circ, \qquad\forall \boldsymbol{K} \in \mathscr{H}\ .
		\end{equation}
	\end{corollary}
	
	\subsection{Haantjes coordinates}
A crucial result for the separability theory of Hamiltonian systems is the theorem proved in \cite{TT2021JGP}, which ensures that an algebra of Haantjes operators attains block-diagonal form in suitable coordinates.

	\begin{proposition} [Proposition 33, \cite{TT2021JGP}] \label{th:HJ}
		Let $\mathcal{K}=\{\boldsymbol{K}_1,\ldots,\boldsymbol{K}_m\}$, $\boldsymbol{K}_\alpha: \mathfrak{X}(M)\to \mathfrak{X}(M)$, $\alpha=1,\ldots,m$ be a family of commuting operator fields; we assume that one of them, say $\boldsymbol{K}_1$, is a Haantjes operator. Then, there exist local charts adapted to the spectral decomposition of $\bs{K}_1$ in which all of the operators $\boldsymbol{K}_\alpha$ can be written simultaneously in a block-diagonal form. In addition, if we assume that
		
		(i) all the operators of the family are Haantjes operators
		
		(ii) all possible nontrivial intersections  of their generalized eigen-distributions
	\begin{equation}
		\label{eq:Va}
		\mathcal{V}_a(\boldsymbol{x}):= \mathcal{D}_{i_1}^{(1)}(\boldsymbol{x}) \bigcap \ldots \bigcap \mathcal{D}_{i_m}^{(m)}(\boldsymbol{x}) \,, \qquad a=1,\ldots, v\leq n,
		\end{equation}
		are integrable, as well as their direct sums, then there exist sets of local coordinates, adapted to the decomposition
		\begin{equation}
			\label{eq:TVa}
			T_{\boldsymbol{x}}M= \bigoplus_{a=1}^{v}\mathcal{V}_a( \boldsymbol{x}), \qquad \boldsymbol{x} \in M,
		\end{equation}
		in which all operators $\boldsymbol{K}_\alpha$ admit simultaneously a block-diagonal form (with possibly finer blocks).
	\end{proposition}
	
	\begin{remark}
		Each set of coordinates allowing us to write down the family $\mathcal{K}$ in a block-diagonal form will be said to be a set of \textit{Haantjes coordinates}. In particular, if we deal with a semisimple Haantjes algebra, then each set of Haantjes coordinates allows us to write simultaneously all the operators of the algebra in a diagonal form.
	\end{remark}
	\begin{definition}
		Given an $\omega \mathscr{H}$ manifold, the local coordinates where all Haantjes operators take simultaneously a block-diagonal form and the symplectic form assumes the Darboux form are called Darboux--Haantjes (DH) coordinates.
	\end{definition}

	\section{The St\"ackel lift as a generalization of the Eisenhart lift} \label{sec:3}

In this section, we introduce and discuss the notion of St\"ackel lift. To this end, we first provide a concise review of the classical Eisenhart and Stäckel frameworks.

	\begin{remark}
		We adopt the following notation: we shall denote by $(q^{i}, p_i)$,  $i=1,\ldots,n$ a generic set of canonically conjugate Darboux coordinates on a $2n$-dimensional symplectic manifold. For Cartesian and polar coordinates, we will keep the standard notations. The new variables introduced in the Eisenhart lifting procedure will be denoted by $(u, p_u)$; for the Lorentzian case, we will use $(t,u)$ and $(p_t, p_u)$.
	\end{remark}
	
	\subsection{The Eisenhart lift} \label{sec:3.1}

	We start by reviewing the Lorentzian and Riemannian Eisenhart lifts \cite{E1928AM,DBKP1985PRD,DGH1991PRD,CKA2015EJP}, which provide a geometric framework for embedding mechanical systems into higher-dimensional spaces as geodesic flows.

	For simplicity, we formulate the theory in the Hamiltonian setting. We consider a symplectic manifold $(M, \omega)$ with symplectic form $\omega= dp_i\wedge dq^{i}$ and a Hamiltonian that is a polynomial of degree 2 in the momenta:
	\begin{equation}
		H= H^{(2)}+H^{(1)}+H^{(0)}.
	\end{equation}

	We focus on the case where
	\begin{equation} \label{eq:3.1}
		H= \frac{1}{2}\sum_{i,j=1}^{n} g^{ij}(\bs{q},t) (p_i -e A_i)(p_j-e A_j) + e^2 V(\bs{q},t),
	\end{equation}
	which includes kinetic energy (quadratic in momenta), a vector potential $\bs{A}$, and a scalar potential $V$.

	To lift this system, we treat time $t$ as a configuration coordinate and introduce a new coordinate $u$. The new momenta $p_u, p_t$ are conjugate to $(u,t)$, and the extended symplectic form becomes
	\begin{equation}
		\omega'=\omega+ dp_t \wedge  dt + dp_u \wedge du.
	\end{equation}
	The lifted Hamiltonian is constructed as
	\begin{equation}
		H^{L}= H^{(2)}+p_u H^{(1)}+ p_u^2 H^{(0)}+ p_t p_u,
	\end{equation}
	which is homogeneous of degree two in momenta. In explicit form,
	\begin{equation} \label{eq:Lor lift}
		H^L= \frac{1}{2}\sum_{i,j=1}^{n} g^{ij}(\bs{q},t) (p_i -p_u A_i)(p_j- p_u A_j) + p_u^2 V(\bs{q},t)+ p_t p_u.
	\end{equation}

	This lifted Hamiltonian generates the geodesic flow of the pseudo-Riemannian metric
	\begin{equation} \label{eq:3.3}
		ds^2=\sum_{i,j=1}^{n} g_{ij}(\bs{q},t)dq^{i} d q^{j}-2V(\bs{q},t)dt^2+2\sum_{i=1}^{n} A_i(\bs{q},t)dq^i dt+2 dt du.
	\end{equation}
	(Here $g^{\mu\nu}g_{\nu\rho}=\delta^\mu_\rho$.) The metric form \eqref{eq:3.3} is a Bargmann–Eisenhart wave \cite{BM2013PRD}.

	The lifting procedure preserves integrability: if $H$ possesses polynomial integrals of motion, the lifted integrals remain constants of motion for $H^L$. 

The momentum $p_u$ is conserved, and the lifted system inherits the integrability properties of the original one. The original system is recovered as the symplectic reduction of the lifted system, fixing the value of its conserved momentum to the charge, $p_u-m=0$; this is the content of the Eisenhart--Lichnerowicz theorem \cite{E1928AM,Lich1955}.
		
	\subsection{Preliminaries on the St\"ackel geometry}\label{sec:3.2}
	Let $(H_1,\ldots, H_n)$ be an integrable Hamiltonian system and assume that all the Hamiltonian functions $H_j$ are quadratic in the momenta 
	\begin{equation} \label{}
		H_j=\frac{1}{2} \sum_{k=1}^n g_{j}^{kk}(\boldsymbol{q}) p_k^2 +V_j(\boldsymbol{q}).
	\end{equation}
	Thus, one can adopt the standard theory of separability due to St\"ackel \cite{St1893CRAS,Pere}. As is well known, St\"ackel's approach provides us with a characterization of \textit{totally separable systems} by means of an invertible \textit{St\"ackel matrix} ${S}$ and \textit{St\"ackel vector} $\bs{W}$, of the form
	\begin{equation} \label{eq:SM}
		S^{(n)}(\bs{q})= 
		\begin{bmatrix}
			S_{11}(q^1)& \cdots & S_{1n}(q^1) \\ 
			\vdots &\ddots & \vdots \\
			S_{n1}(q^n) & \cdots & S_{nn}(q^n)
		\end{bmatrix} \, ,
		\qquad\qquad 
		\bs{W}= 
		\begin{bmatrix}
			W_{1}(q^1) \\ 
			\vdots \\
			W_{n}(q^n) 
		\end{bmatrix} \,,
	\end{equation} 
	where $S_{ij}$ and $W_{i}$ are arbitrary functions of their separated arguments that satisfy the conditions
	\begin{equation}\label{eq:stackel_data}
		g_{j}^{kk}(\boldsymbol{q})=[(S^{(n)})^{-1}]_{j}^{k}(\bs{q}) \, , \qquad\qquad V_{j}(\bs{q})=\sum_{i=1}^n g^{ii}_j(\boldsymbol{q})\, W_i(q^i) \ .
	\end{equation}
	Thus, we have
	\begin{equation} \label{eq:1.9}
		H_j=\sum_{k=1}^n [(S^{(n)})^{-1}]_{j}^{k}(\boldsymbol{q})\Big( \frac{p_k^2}{2} +W_{k}(q^k)\Big) \,.
	\end{equation}

	In \cite{AKN1997}, Arnold, Kozlov and Neishtadt (AKN) proposed a generalization of St\"ackel systems by considering the case of arbitrary, not necessarily quadratic, Hamiltonian functions defined in a symplectic manifold. Precisely, given a St\"ackel matrix $\bs{S}$, their construction consists in the family of Hamiltonians
	\begin{equation}
		H_{j} = \sum_{k=1}^{n} \dfrac{\tilde{S}_{j}^{k}}{\det S^{(n)}} \, f_{k}(q^k,p_k), \qquad j=1,\ldots, n,
		\label{AKN}
	\end{equation}
	where $\tilde{S}_{j}^{k}$ is the cofactor associated to the element $S_{k}^{j}$, and $f_{k}$ are arbitrary functions forming a \emph{generalized St\"ackel vector} $\boldsymbol{f}$. It can be shown that the functions $H_j$, defined by Eqs. \eqref{AKN}, are in involution with each other; they represent a totally separable Hamiltonian system.
	
	\subsection{The St\"ackel lift}\label{sec:3.3}

	\begin{definition} Let $Q$ be an $n$-dimensional manifold, with local coordinates $(q^{1},\ldots,q^{n})$. Define $\tilde{Q}= Q\times\mathbb{R}$, where $q^{n+1}$ denotes the coordinate on $\mathbb{R}$, and let $T^{*}\tilde{Q}$ be its cotangent bundle.
		Let $S^{(n)}(\bs{q})$ be an $(n\times n)$ St\"ackel matrix.
		The Stäckel lift of this matrix is defined to be the matrix
		\begin{equation} \label{eq:LSM}
			\mathcal{L}^{(n+1)}(\bs{q}, \bs{p}):= \left[\begin{array}{ccc|ccc}
				S_{11}(q^1)& \cdots & S_{1n}(q^1) & S_{1,n+1}(q^1,p_1) \\ 
				\vdots &\ddots & \vdots & \vdots \\
				S_{n1}(q^n) & \cdots & S_{nn}(q^n) &S_{n,n+1}(q^n,p_n)\\ \hline
				0 & \cdots & 0 & S_{n+1,n+1}(q^{n+1},p_{n+1})
			\end{array}\right].
		\end{equation}
		Here $S_{1,n+1}(q^{1},p_1),\ldots,S_{n+1,n+1}(q^{n+1},p_{n+1})$ are arbitrary smooth functions, with $S_{n+1,n+1}(q^{n+1},p_{n+1})\neq 0$. 
	\end{definition}
	\begin{remark}
	In the following, $(T^{*}\tilde{Q}, \tilde{\omega})$ will be regarded as a symplectic manifold endowed with the canonical symplectic form $\tilde{\omega}= \sum_{i=1}^{n+1} dp_i\wedge dq^i$.
	\end{remark}
	\begin{remark}
		The lifted matrix $\mathcal{L}^{(n+1)}(\bs{q}, \bs{p})$ depends explicitly, in its last column, on the momenta $\bs{p}=(p_1,\ldots,p_{n+1})$. In Remark 6 of \cite{RTT2024AMPA} it was observed that, interestingly enough, even when each row $a$ of the Stäckel matrix $\boldsymbol{S}$ depends on $(q^a, p_a)$, we can still construct a separable Hamiltonian system in involution. This fact also holds for partially separable systems. This observation will be crucial in the subsequent analysis.
	\end{remark}
	
	\begin{remark}
		In the case when no dependence on the momenta is allowed, the lifted matrix $\mathcal{L}^{(n+1)}(\bs{q}, \bs{p})$ is by construction a standard St\"ackel matrix on an $(n+1)$-dimensional configuration space. There are many other possible choices of lifting a St\"ackel matrix by preserving its St\"ackel character; the present one is perhaps the simplest one, and is sufficient to describe many interesting cases.
	\end{remark}
	
	\subsection{Integrable Hamiltonian systems from the St\"ackel lift}
	Let $(H_1,\ldots,H_n)$ be an integrable Hamiltonian system, separable in the coordinate system $\bs{q}:=(q^1,\ldots,q^n)$, and $S^{(n)}(\bs{q})$ be its associated St\"ackel matrix.
	The lifted-St\"ackel Hamiltonian system on $T^{*}\tilde{Q}$ associated with the matrix \eqref{eq:LSM} reads

	\begin{equation} \label{eq:LSH}
		\begin{bmatrix}
			H_1 \\
			\vdots \\
			H_{n+1}
		\end{bmatrix}
		=
		(\mathcal{L}^{(n+1)})^{-1}(\bs{q}, \bs{p})
		\begin{bmatrix}
			f_1 (q^1,p_1)\\
			\vdots \\
			f_{n+1}(q^{n+1},p_{n+1})
		\end{bmatrix}.
	\end{equation}
	
	 The general form of the time-independent Hamilton-Jacobi equations is
\begin{equation*}
H_i \left( q^{1}, ..., q^{n+1}, p_{1} = \dfrac{\partial W}{\partial q^{1}}, ..., p_{n+1} = \dfrac{\partial W}{\partial q^{n+1}} \right) = h_i, \qquad i=1,\ldots,n+1
\end{equation*}
where $h_i$ are arbitrary constants, and the complete integral $W$ satisfies the condition 
$
\det \bigg[ \frac{\partial^2 W}{\partial q^i\partial h_j} \bigg]\neq 0 \ 
$. The HJ equations  split into the system of separation equations
	\begin{equation} \label{eq:5.10}
		f_i\left(q^i, \frac{\partial W} {\partial q^i}\right)-\langle \mathcal{L}^{(n+1)}_i~|~\bs{h}~\rangle=0 \qquad \qquad i=1,\ldots,n+1 \,
	\end{equation}
	where $\mathcal{L}^{(n+1)}_i$ is the $i$-th row of the lifting matrix $\mathcal{L}^{(n+1)}(\bs{q},\bs{p})$, and $\bs{h}=[h_1,\ldots,h_{n+1}]^{T}$, is a vector of arbitrary constants.

	\subsubsection{Lifted Systems separable in 3D}
	For concreteness, we illustrate the Stäckel lift for a separated Hamiltonian system with $n=2$ and Stäckel matrix $S^{(2)}(\bs{q})$. Let us consider the system
	\begin{equation} \label{eq:GSS3}
		\begin{bmatrix}
			H_1 \\
			H_2\\
			H_3
		\end{bmatrix}
		=(\mathcal{L}^{(3)})^{-1}(\bs{q}, \bs{p})
		\begin{bmatrix}
			f_1 (q^{1}, p_{1})\\
			f_2(q^{2}, p_{2})\\
			f_3(q^{3}, p_{3})
		\end{bmatrix}\,,
	\end{equation}
	where $f_{1}$, $f_{2}$ and $f_{3}$ are arbitrary functions and the lifted Stäckel matrix \eqref{eq:LSM} takes the form
	\begin{equation} \label{eq:3.12}
		\mathcal{L}^{(3)}(\bs{q}, \bs{p}) =
		\begin{bmatrix}
			S_{11} (q^{1}) & S_{12} (q^{1}) & S_{13} (q^{1},p_1) \\
			S_{21} (q^{2}) & S_{22} (q^{2}) & S_{23} (q^{2},p_2) \\
			0 & 0 & S_{33} (q^{3},p_3) 
		\end{bmatrix},
	\end{equation}
	with $S_{33}(q^3,p_3)\neq 0$ to ensure that $\det \mathcal{L}^{(3)}\neq 0$. From eq. \eqref{eq:GSS3}, one can obtain a large family of Hamiltonian systems $(H_1, H_2, H_3)$ in involution. It explicitly reads
	\begin{equation}
		\begin{split}
			H_{1} &= \dfrac{1}{\det \mathcal{L}^{(3)}} \left[ S_{22} S_{33} f_{1} - S_{12} S_{33} f_{2} + ( S_{12} S_{23} - S_{13} S_{22} ) f_{3} \right], \\
			H_{2} &= \dfrac{1}{\det \mathcal{L}^{(3)}} \left[ - S_{21} S_{33} f_{1} + S_{11} S_{33} f_{2} + ( S_{13} S_{21} - S_{11} S_{23} ) f_{3} \right], \\
			H_{3} &= \dfrac{1}{\det \mathcal{L}^{(3)}} \left[ ( S_{11} S_{22} - S_{12} S_{21} ) f_{3} \right].
		\end{split}
	\end{equation}
	From the previous relations we obtain the set of three associated separation equations: 
	\begin{equation}
		\begin{cases}
			f_1\left(q^1, \frac{\partial W} {\partial q^1} \right)- S_{11} h_1-S_{12} h_2-S_{13}h_3=0, \\
			f_2\left(q^2, \frac{\partial W} {\partial q^2}\right)- S_{21} h_1-S_{22}h_2-S_{23}h_3 =0,\\
			f_3\left(q^3, \frac{\partial W} {\partial q^3}\right)- S_{33}h_3=0.
		\end{cases}
	\end{equation}

	\begin{remark}\label{rmk momentum dep}
		Equation \eqref{eq:LSH} defines Hamiltonian systems that depend on momenta not only through the St\"ackel functions $f_k(q^k, p_k)$, but also through the coefficients of the inverse of the momentum-dependent matrix \eqref{eq:LSM}. From a geometrical point of view, this corresponds to imposing a momentum-dependent metric on $T^*\tilde{Q}$. These are known as Hamilton and generalized Hamilton spaces \cite{Miron2001}; see, e.g., \cite{RSS2024} for a recent mathematical treatment and further references. Generally speaking, one can prove that the Hamiltonians defined by means of the St\"ackel lift are functionally independent in a dense, open subset of the extended phase space.
	\end{remark}
	
	\subsection{The Riemannian Eisenhart lift as a St\"ackel lift}
	As a first example of the previous approach, we discuss the standard Eisenhart lift, in its Riemannian version, for a class of integrable models. 
Precisely, we will show that the Eisenhart lifting procedure is equivalent to a St\"ackel lift of a suitable St\"ackel matrix defined on the configuration space, according to the following result.

\begin{theorem}[Stäckel lift and Riemannian Eisenhart lift]\label{theo:1}
	Let
	\begin{equation}\label{eq:H}
		H_j = \frac{1}{2}\sum_{i=1}^n g^{ii}_j(\boldsymbol{q})\, p_i^2 + V_j(\boldsymbol{q}), \qquad j=1,\ldots,n,
	\end{equation}
	be a Stäckel-separable completely integrable Hamiltonian system on $T^{*}Q$, separable in a coordinate chart
	$(q^1,\ldots,q^n)$ with associated Stäckel matrix $S^{(n)}$.
	Let the contravariant metric and Stäckel potentials be given by eqs.
	\eqref{eq:stackel_data}.
	Consider the lifting matrix
	\begin{equation}\label{eq:lifting-matrix}
		\mathcal{L}^{(n+1)}(\boldsymbol{q})=
		\begin{pmatrix}
			S^{(n)}(\boldsymbol{q}) & -\boldsymbol{W} \\
			\boldsymbol{0}^T & 1
		\end{pmatrix},
		\qquad
		\boldsymbol{W}=(W_1(q^1),\ldots,W_n(q^n))^T .
	\end{equation}
	Then, choosing the Stäckel functions $f_i = \frac{1}{2}p_i^2$ (for $i=1,\ldots,n$) and $f_{n+1} = p_u^2$, the Stäckel lift generated by $\mathcal{L}^{(n+1)}$ yields the independent lifted Hamiltonians on $T^{*}\tilde{Q}$, making the system completely integrable: 
	\begin{equation}\label{eq:lifted-H}
		\tilde H_j
		= \frac{1}{2}\sum_{i=1}^n g^{ii}_j(\boldsymbol{q})\, p_i^2
		+ V_j(\boldsymbol{q})\, p_u^2, \quad j=1,\ldots,n,\qquad \tilde{H}_{n+1}=p_u^2.
	\end{equation}
	
	
In particular, $\tilde H_1$ is the geodesic Hamiltonian of the Riemannian Eisenhart lift of the natural system with metric $g_{ij}$ and potential $V_1$ (Riemannian wherever $V_1>0$), while the remaining $\tilde H_j$ are higher Stäckel integrals corresponding to Killing tensors of the same lifted metric.
\end{theorem}

\begin{proof}
	For Stäckel-separable systems, the contravariant metric and potentials admit
	the representations
	\begin{equation}\label{eq:metric-stackel}
		g^{kk}_j(\boldsymbol{q})=\bigl[(S^{(n)})^{-1}\bigr]_{j}^{k},
	\end{equation}
	and
	\begin{equation}\label{eq:potential-stackel}
		V_j(\boldsymbol{q})=\sum_{i=1}^n g^{ii}_j(\boldsymbol{q})\, W_i(q^i),
	\end{equation}
	see \cite{Benenti1997}.
	
	Since $S^{(n)}$ is invertible, the inverse of the lifting matrix
	\eqref{eq:lifting-matrix} is obtained via the Schur complement.
	For a block matrix
	\begin{equation}\label{eq:block-matrix}
		M=\begin{pmatrix} A & B \\ C & D \end{pmatrix},
	\end{equation}
	where $A$ is invertible and whose Schur complement
	\begin{equation}\label{eq:schur-complement}
		S_D = D - C A^{-1} B
	\end{equation}
	is also invertible, the inverse of $M$ is given by
	\begin{equation}\label{eq:schur-inverse}
		M^{-1}
		=
		\begin{pmatrix}
			A^{-1} + A^{-1} B S_D^{-1} C A^{-1}
			& -A^{-1} B S_D^{-1} \\
			- S_D^{-1} C A^{-1}
			& S_D^{-1}
		\end{pmatrix}.
	\end{equation}
	
	In the present case $A=S^{(n)}$, $B=-\boldsymbol{W}$,
	$C=\boldsymbol{0}^T$, and $D=1$, so that $S_D=1$. Substitution into
	\eqref{eq:schur-inverse} yields
	\begin{equation}\label{eq:lifted-inverse}
		[\mathcal{L}^{(n+1)}]^{-1}
		=
		\begin{pmatrix}
			(S^{(n)})^{-1} & (S^{(n)})^{-1}\boldsymbol{W} \\
			\boldsymbol{0}^T & 1
		\end{pmatrix}.
	\end{equation}
	
	The Stäckel lift defines a separable Hamiltonian system on $T^*\tilde Q$,
	$\dim\tilde Q=n+1$. The $j$-th lifted Hamiltonian is
	\begin{equation}\label{eq:H-tilde-def}
		\tilde H_j
		= \frac{1}{2}\sum_{i=1}^n [\mathcal{L}^{-1}]_{ji}\, p_i^2
		+ [\mathcal{L}^{-1}]_{j,n+1}\, p_u^2 .
	\end{equation}
	
	From \eqref{eq:lifted-inverse} and \eqref{eq:metric-stackel} we have
	\begin{equation}
		[\mathcal{L}^{-1}]_{ji}=g^{ii}_j(\boldsymbol{q}),
	\end{equation}
	while using \eqref{eq:potential-stackel} gives
	\begin{equation}
		[\mathcal{L}^{-1}]_{j,n+1}
		= \sum_{k=1}^n g^{kk}_j(\boldsymbol{q})\, W_k(q^k)
		= V_j(\boldsymbol{q}).
	\end{equation}
	Substituting into \eqref{eq:H-tilde-def} yields \eqref{eq:lifted-H}, completing
	the proof.
\end{proof}

\subsection{The St\"ackel lift of Smorodinsky-Winternitz systems}
As an application of Theorem \ref{theo:1}, we shall revise the natural Hamiltonian system, separable in Cartesian coordinates
	\begin{equation} \label{eq:3.18}
		\begin{split}
			H_1 &= \frac{1}{2}\big(p_1^2+p_2^2+V_1(x^1)+V_2(x^2)\big), \\
			H_2 &= p_1^2+V_1(x^1).
		\end{split}
	\end{equation}The St\"ackel-lifted matrix
	\begin{equation} \label{Stackel 1}
		\mathcal{L}^{-1}= 
		\begin{bmatrix}
			1 & 1 & V_1(x^1)+V_2(x^2) \\ 
			2 & 0 & 2 V_1(x^1) \\ 
			0 & 0 & 1
		\end{bmatrix}, \qquad \mathcal{L}= 
		\begin{bmatrix}
			0 & 1/2 & -V_1(x^1) \\ 
			1 & -1/2 & -V_2(x^2) \\ 
			0 & 0 & 1
		\end{bmatrix},
	\end{equation}
	allows us to define the Eisenhart lift of the family of systems \eqref{eq:3.18}. Here the St\"ackel functions in \eqref{AKN} are just quadratic terms in momenta multiplied by a constant. We obtain
	\begin{equation} \label{eq:Lor lift1}
		\begin{split}
			H_1 &= \frac{1}{2}\Big(p_1^2+p_2^2+(V_1(x^1)+V_2(x^2))p_u^2\Big), \\
			H_2 &= p_1^2+\Big(V_1(x^1)\Big)p_u^2,\\
			H_3 &= p_2^2+\Big(V_2(x^2)\Big)p_u^2,\\
			H_4 &= \frac{1}{2}p_u^2,
		\end{split}
	\end{equation}
	where the integral $H_3=2 H_1-H_2$. 
The class of systems \eqref{eq:Lor lift1} contains the Eisenhart lifts of the Smorodinsky-Winternitz systems I and II, separating in Cartesian coordinates, which have been obtained by Cari\~nena et al.  in \cite{Carinena2017}. A completely analogous analysis can be performed for the Smorodinsky-Winternitz systems III, separating in polar and parabolic coordinates, and IV, which separates in two different parabolic coordinate systems.

\begin{remark}[Conformal Eisenhart lifts and position-dependent mass]\label{rmk:conformal}
	The Stäckel framework naturally extends to conformal generalizations of the Eisenhart lift studied by Cariñena, Herranz, and Rañada \cite{Carinena2017}. Consider the generalized lifted matrix
	\begin{equation}
		\mathcal{L}^{(n+1)}(\mathbf{q}) = \begin{pmatrix}
			S^{(n)}(\mathbf{q}) & -\mathbf{W}(\mathbf{q}) \\
			\mathbf{C}(\mathbf{q})^T & 1
		\end{pmatrix},
	\end{equation}
	where $\mathbf{C}(\mathbf{q})$ is a suitable $n$-dimensional vector and $\boldsymbol{W}$ is the St\"ackel vector \eqref{eq:stackel_data}, ensuring invertibility of $\mathcal{L}^{(n+1)}$.

	Using the expression \eqref{eq:schur-inverse} for the matrix inverse obtained via the Schur complement, and choosing $\bs{C}$ conveniently, one obtains for the first lifted Hamiltonian the form
	\begin{equation}
		\tilde{H}_1
		= \frac{1}{\Omega(\boldsymbol{q})}
		\left(
		\frac{1}{2}\sum_{i=1}^n g^{ii}(q^i)\, p_i^2
		+ V(q)\, p_u^2
		\right),
	\end{equation}
	i.e., a conformal Riemannian Eisenhart lift with a position-dependent conformal
	factor $\Omega$ corresponding to the Schur complement
	$\Omega(\boldsymbol{q}) = 1 + \boldsymbol{C}(\boldsymbol{q})^{T}
	\bigl(S^{(n)}(\boldsymbol{q})\bigr)^{-1}\boldsymbol{W}(\boldsymbol{q})$.

	For example, with $n=2$ and
	\begin{equation}
		\mathcal{L} = \begin{pmatrix} 
			0 & \frac{1}{2} & -V_1(x) \\
			1 & -\frac{1}{2} & -V_2(y) \\
			C & 0 & 1 
		\end{pmatrix},
	\end{equation}
	the conformal factor becomes $\Omega = 1 + C[V_1(x) + V_2(y)]$, coupling all three coordinates. Setting $V_i = \frac{\omega^2}{2}(x^i)^2$ and $C = -\lambda$ recovers the Darboux III oscillator with $\Omega \propto 1 - \lambda r^2$ (polar coordinates), while an analogous construction in parabolic-cylindrical coordinates with Coulomb-type potentials and $C = -\lambda$ yields Taub-NUT related factors $\Omega \propto 1 - \lambda/r$.
	
	Thus, the Stäckel framework encodes both standard Eisenhart lifts (Theorem \ref{theo:1}) and their conformal generalizations, with all four superintegrable systems in \cite{Carinena2017} being realizable as special, multiseparable cases.
\end{remark}

	\section{Haantjes geometry and St\"ackel lifts} \label{sec:4}
	
	In this section, we wish to clarify the relationship between the St\"ackel lift and the Haantjes geometry.
	
	\subsection{$\omega\mathscr{H}$ manifolds and St\"ackel lifts}
	\begin{theorem} \label{maintheorem}
		Let $\big(H_1,\ldots,H_{n+1}\big)$ be a St\"ackel-lifted Hamiltonian system defined by relations \eqref{eq:LSH}, for a lift matrix of the form
		\begin{equation} 
			\mathcal{L}^{(n+1)}(\bs{q}, \bs{p}):= \left[\begin{array}{ccc|ccc}
				S_{11}(q^1)& \cdots & S_{1n}(q^1) & S_{1,n+1}(q^1,p_1) \\ 
				\vdots &\ddots & \vdots & \vdots \\
				S_{n1}(q^n) & \cdots & S_{nn}(q^n) &S_{n,n+1}(q^n,p_n)\\ \hline
				0 & \cdots & 0 & S_{n+1,n+1}(q^{n+1})
			\end{array}\right],
		\end{equation}
		with $ S_{n+1,n+1}(q^{n+1})\neq 0$. Then, there exists an associated $\omega \mathscr{H}$ manifold, whose Haantjes algebra $\mathscr{H}$ is defined by the set of semisimple operators
		\begin{equation} \label{eq:HaaK}
			\boldsymbol{K}_{j}:=\sum_{r=1}^{n+1} \mu_{jr}\bigg( \frac{\partial}{\partial q^r}\otimes \mathrm{d} q^r+ \frac{\partial}{\partial p_r}\otimes \mathrm{d} p_r\bigg),
			\quad \ j=1,\ldots,n+1,
		\end{equation}
		where $\mu_{jr}$ has the form
		\begin{equation} \label{eq:mu}
			\mu_{jr}\big(q^1,\ldots, q^n,p_1,\ldots,p_n\big):=\dfrac{\tilde{\mathcal{L}}_{jr }}{\tilde{\mathcal{L}}_{1r}}.
		\end{equation}
Here $\tilde{\mathcal{L}}_{jr}$ denotes the cofactor associated with the element $\mathcal{L}_{rj}$ of the lifting matrix, and we restrict to the open subset of the extended phase space where $\tilde{\mathcal{L}}_{1r}$ is nonzero, $r = 1,\ldots, n+1$ and $H_1$ depends on all the momenta.
	\end{theorem}
	\begin{proof}
		As a consequence of the Jacobi-Haantjes theorem, proved in \cite{RTT2022CNS}, the Hamiltonians $\big(H_1,H_2,\ldots,H_{n+1}\big)$ which are in total separable involution in the Darboux coordinates $(\bs{q},\bs{p})$, belong to the Haantjes chain generated by the operators 
		\begin{equation} \label{eq:LSoV}
			\boldsymbol{K}_j=\sum _{i=1}^{n+1} \frac{\frac{\partial H_{j}}{\partial p_i}}{ \frac{\partial H}{\partial p_i}}\bigg (\frac{\partial}{\partial q^i}\otimes \rd q^i +\frac{\partial}{\partial p_i}\otimes \rd p_i \bigg ),\qquad j=1,\ldots,n+1,
		\end{equation}
		where $H$ is any of the functions $(H_1, \ldots, H_{n+1})$, with $\frac{\partial H}{\partial p_i}\neq 0$, $i=1,\ldots,n+1$.
		
		Taking into account that the Hamiltonian system is of the AKN form \eqref{AKN}, one can observe that
		\begin{align}
			\nn \frac{\frac{\partial H_{j}}{\partial p_r}}{ \frac{\partial H_1}{\partial p_r}} ={}& \dfrac{\dfrac{1}{\det \mathcal{L}^{(n+1)}} \left(\sum_{k=1}^{n} \tilde{\mathcal{L}}_{jk}\, \frac{\partial f_{k}(q^k, p_k)}{\partial{p_r}}+ \frac{\partial \tilde{\mathcal{L}}_{j,n+1}}{\partial p_r} f_{n+1}(q^{n+1}, p_{n+1})\right)}{\dfrac{1}{\det \mathcal{L}^{(n+1)}} \left(\sum_{k=1}^{n} \tilde{\mathcal{L}}_{1k}\, \frac{\partial f_{k}(q^k, p_k)}{\partial{p_r}}+ \frac{\partial \tilde{\mathcal{L}}_{1,n+1}}{\partial p_r} f_{n+1}(q^{n+1}, p_{n+1})\right)}= 	\\
			{}&  \dfrac{\tilde{\mathcal{L}}_{jr }\frac{\partial f_r}{\partial p_r}+\frac{\partial \tilde{\mathcal{L}}_{j,n+1}}{\partial p_r}f_{n+1}(q^{n+1}, p_{n+1})}{\tilde{\mathcal{L}}_{1r}\frac{\partial f_r}{\partial p_r}+\frac{\partial \tilde{\mathcal{L}}_{1,n+1}}{\partial p_r}f_{n+1}(q^{n+1}, p_{n+1}) }. 
			\end{align}
Taking into account that $\tilde{\mathcal{L}}_{jr}=S_{n+1,n+1}(q^{n+1})\tilde{S}_{jr}^{(n)}$, $j,r\leq n$, and $\tilde{\mathcal{L}}_{j,n+1}=-\sum_{i=1}^{n}\tilde{S}^{(n)}_{ji}S_{i,n+1}(q^i,p_i)$, we arrive at the final expression
\[
\mu_{jr}\big(q^1,\ldots, q^n,p_1,\ldots,p_n\big) = \dfrac{\tilde{\mathcal{L}}_{jr }}{\tilde{\mathcal{L}}_{1r}}, \qquad j,r=1,\ldots, n+1.
\]
For $r=n+1$, only $f_{n+1}$ depends on $p_{n+1}$, and the same conclusion follows. These operators, being diagonal, have vanishing Haantjes torsion and generate a semisimple Abelian $\omega \mathscr{H}$ structure on $T^{*}\tilde{Q}$.
	\end{proof}
	\begin{remark}
		This theorem generalizes Proposition 3.1, proved in \cite{TT2016SIGMA}, which establishes the existence of a Haantjes structure associated with the standard AKN systems \eqref{AKN}. However, there is a crucial difference: the Haantjes operators \eqref{eq:HaaK}, unlike the case of those associated with systems \eqref{AKN}, are \textit{non-projectable} along the fibers of $T^{*}\tilde{Q}$ into operators with vanishing Haantjes torsion on the configuration space $\tilde{Q}$. This distinguishes momentum-dependent lifts from purely geodesic lifts and represents a genuinely new class of Haantjes structures.
	\end{remark}
	
	\subsection{Haantjes algebras for momentum-dependent Stäckel lifts} 
	
	To illustrate the main result of our construction—that Stäckel lifts with momentum-dependent coefficients generate Haantjes algebras where the operators themselves depend on momenta—consider the general case where the lifting functions $V_1$ and $V_2$ depend on both configuration and momentum variables: $V_1 = V_1(x^1, p_1)$ and $V_2 = V_2(x^2, p_2)$.
	
	For the algebra $\mathcal{A}_2=\{H_1,H_3, H_4 \}$ constructed from the deformed system \eqref{eq:3.18}, the associated lifting matrix reads
	\begin{equation} \label{Stackel 2}
		\mathcal{L}^{-1}= 
		\begin{bmatrix}
			1 & 1 & V_1(x^1, p_1)+V_2(x^2, p_2) \\ 
			0 & 2 & 2 V_2(x^2, p_2) \\ 
			0 & 0 & 1
		\end{bmatrix}, \qquad \mathcal{L}= 
		\begin{bmatrix}
			1 & -1/2 & -V_1(x^1, p_1) \\ 
			0 & 1/2 & -V_2(x^2, p_2) \\ 
			0 & 0 & 1
		\end{bmatrix}.
	\end{equation}
	
	By using eqs. \eqref{eq:HaaK}-\eqref{eq:mu}, after simple algebraic manipulations we obtain the Haantjes operators
	\begin{equation}\label{K1 cart mom}
		\bs{K}_1:=\left[\begin{array}{cccccc}
			2 & 0 & 0 & 0 & 0 & 0 \\
			0 & 0 & 0 & 0 & 0 & 0 \\
			0 & 0 & \frac{2 V_1(x^1, p_1)}{V_1(x^1, p_1)+V_2(x^2, p_2)} & 0 & 0&0 \\
			0 & 0 & 0 & 2 & 0 & 0 \\
			0 & 0 & 0 & 0 & 0 & 0 \\
			0 & 0 & 0 & 0 & 0 & \frac{2 V_1(x^1, p_1)}{V_1(x^1, p_1)+V_2(x^2, p_2)}
		\end{array}\right],
	\end{equation}
	
	\begin{equation}\label{K2 cart mom}
		\bs{K}_2:=\left[\begin{array}{cccccc}
			0 & 0 & 0 & 0 & 0 & 0 \\
			0 & 2 & 0 & 0 & 0 & 0 \\
			0 & 0 & \frac{2 V_2(x^2, p_2)}{V_1(x^1, p_1)+V_2(x^2, p_2)} & 0 & 0 &0\\
			0 & 0 & 0 & 0 & 0 & 0 \\
			0 & 0 & 0 & 0 & 2 & 0 \\
			0 & 0 & 0 & 0 & 0 & \frac{2 V_2(x^2, p_2)}{V_1(x^1, p_1)+V_2(x^2, p_2)}
		\end{array}\right],
	\end{equation}
	
	\begin{equation}
		\bs{K}_3=\frac{1}{V_1(x^1, p_1)+V_2(x^2, p_2)}\mathrm{diag}(0,0,1,0,0,1).
	\end{equation} 
	
Here $\bs{K}_1$ is the Haantjes operator associated with $H_2=2H_1-H_3$. The momentum dependence enters explicitly through the ratios
	
	$$\frac{2V_1(x^1, p_1)}{V_1(x^1, p_1)+V_2(x^2, p_2)} \quad \text{and} \quad \frac{2V_2(x^2, p_2)}{V_1(x^1, p_1)+V_2(x^2, p_2)}$$
	in both the $(3,3)$ and $(6,6)$ entries of $\bs{K}_1$ and $\bs{K}_2$, respectively, and through the factor $1/(V_1(x^1, p_1)+V_2(x^2, p_2))$ in the diagonal entries of $\bs{K}_3$. This shows that the Haantjes algebra of operators $\{\bs{K}_1, \bs{K}_2, \bs{K}_3\}$ depends on the phase-space variables and is therefore non-projectable.
	
	\begin{remark} When the functions $V_1$ and $V_2$ depend only on variables of the configuration space, i.e., $V_1 = V_1(x^1)$ and $V_2 = V_2(x^2)$, the momentum dependence disappears from the Haantjes operators. In this case, we recover the standard Riemannian Eisenhart lift discussed in Section~\ref{sec:3}, where the algebra $\mathcal{A}_1=\{H_1,H_2, H_4 \}$ can be regarded as a geodesic Hamiltonian system in three dimensions with 
		\begin{gather}
		g_1^{kk}(\bs{q})= \mathrm{diag}(1,1, V_1(x^1)+V_2(x^2)),\\
		g_2^{kk}(\bs{q})= \mathrm{diag}(2,0, 2~V_1(x^1)), \qquad 
		g_4^{kk}(\bs{q})= \mathrm{diag}(0,0, 1),
	\end{gather}
	and the associated Haantjes operators induce configuration-space Killing tensors.
\end{remark}
	\section{New separable systems generated by iterated geodesic lifts} \label{sec:5}

	In the spirit of the Eisenhart lift's embedding of dynamics as higher-dimensional geodesic motion, we propose a lift that preserves this geodesic character. We achieve this by first analyzing one of the simplest realizations of the momentum-dependent St\"ackel lift \eqref{eq:LSM}-\eqref{eq:LSH}, given by 
	
	\begin{equation} \label{eq:4.1}
		\mathcal{L}^{(3)}(\bs{q},\bs{p})=
		\begin{bmatrix}
			0 & 1/2 & -V_1(x^1) p_1 \\ 1 & -1/2 & -V_2(x^2)p_2 \\ 0 & 0 & 1
		\end{bmatrix}, 
	\end{equation}
	where $\tbf{q}=(x^1,x^2,u)$ denotes our configuration coordinates, and the functions $V_i(x^{i})$ come from the additively separated potential of the system \eqref{eq:3.18}. This corresponds to lifting in a novel, non-Eisenhart way.
	By choosing $f_1=\frac{1}{2}p^{2}_{1}$, $f_2=\frac{1}{2}p^{2}_{2}$, $f_3= p_u$, and applying the lift to the 2D model \eqref{eq:3.18}, we obtain the family of integrable systems, defining a geodesic motion in 3D\footnote{Hereafter we will scale the resulting $H_2$ by a half, and take the square of $H_3$.}:
	
	\begin{equation} \label{Stackel Riem lift}
		\begin{split}
			H_1 &= \frac{1}{2}(p_1^2+p_2^2) + V_1(x^1) p_1 p_u+ V_2(x^2) p_2 p_u, \\ 
			H_2 &= \frac{1}{2}p_1^2+ V_1(x^1) p_1 p_u, \\ 
			H_3 &= \frac{1}{2}p_u^2.
		\end{split}
	\end{equation}
	
By adopting the alternative St\"ackel basis $(\tilde{H}_1=H_1+H_3,\tilde{H}_2=H_2,\tilde{H}_3=H_3)$, chosen so that $\tilde{H}_1$ possesses nonvanishing diagonal metric components $g^{ii}_1\neq0$, the Hamiltonian system $(\tilde{H}_1,\tilde{H}_2,\tilde{H}_3)$ could be interpreted, at least formally, as a ``magnetic'' model arising from the Riemannian Eisenhart lift (compare with \eqref{eq:Lor lift}),
	\begin{equation}
		H= \frac{1}{2}\sum_{i,j=1}^{n} g^{ij}(\bs{q},t) (p_i -p_u V_i)(p_j- p_u V_j) + p_u^2 V(\bs{q},t),
	\end{equation}
	with an associated vector potential $A=(-V_1(x^1),-V_2(x^2),0)$.
	However, if we rewrite the system in the gauge-covariant form
	\begin{equation}
		\tilde{H}_1= \frac{1}{2}\left[(p_1+V_1(x^1) p_u)^2+(p_2+V_2(x^2)p_u)^2+(1-V_1(x^1)^2-V_2(x^2)^2)p_u^2\right],
	\end{equation}
	we observe that $B=\operatorname{curl} A=0$. Thus, there is no  term depending on $p_1$, $p_2$ in the force
	$F_i=V_i(x^i) \partial_{i}V_i(x^i)p_u^2$. By means of the canonical transformation
	\begin{equation}\label{gauge riem lift}
		(Q_1,Q_2,Q_3,P_1,P_2,P_3)=(x^1,x^2,u-\int V_1(x^1)\mathrm{d}x^1-\int V_2(x^2)\mathrm{d}x^2,p_1+V_1(x^1) p_u,p_2+V_2(x^2)p_u,p_u),
	\end{equation}
	we are led to the ``nonmagnetic'' form
	\begin{equation}
		\tilde{H}_1= \frac{1}{2}\left[P_1^2+P_2^2+(1-V_1(Q_1)^2-V_2(Q_2)^2)P_3^2\right].
	\end{equation}
	The canonical transformation \eqref{gauge riem lift} can be regarded as the \textit{Riemannian lift of the gauge transformation}.\footnote{Geometrically, according to \cite{BM2013PRD}, such transformations correspond to different but equivalent choices of the ``screen'' world volumes that foliate the space(time).}

	The lift \eqref{eq:4.1} can be naturally iterated to generate higher-dimensional Hamiltonian integrable systems. As a second step, in the cotangent bundle $T^{ *}\tilde{Q}$ with coordinates $(x^1,x^2,u_1,u_2, p_1,p_2,p_{u_1},p_{u_2})$, we define the lifting matrix
	\begin{equation}
		\mathcal{L}^{(4)}(\bs{q},\bs{p}):=
		\begin{bmatrix}
			0 & 1/2 & -V_1(x^1) p_1 & -W_{1}(x^1) p_1 \\ 1 & -1/2 & -V_2(x^2)p_2 & -W_{2}(x^2) p_2 \\ 0 & 0 & 1 & -W_{3}(u_1) \\ 0 & 0 & 0 & 1
		\end{bmatrix},
	\end{equation}
	where $W_{1}(x^1)$, $W_{2}(x^2)$, $W_{3}(u_1)$ are arbitrary functions.
	By choosing as St\"ackel functions $f_1= \frac{1}{2}p_1^2$, $f_2=\frac{1}{2}p_2^2$, $f_3= p_{u_1}$, $f_4= p_{u_2}$, we obtain the lifted, homogeneous of degree 2 and 4D integrable model (again scaling $H_2$ by a half and squaring $H_3$ and $H_4$)
	\begin{equation}
		\begin{split}
			H_1 &= \frac{1}{2}\big(p_1^2+p_2^2\big) + V_1(x^1) p_1 p_{u_1} + V_2(x^2) p_2 p_{u_1} +\big(W_{1}(x^1)+ V_1(x^1) W_{3}(u_1) \big) p_1 p_{u_2} \\
			&\quad +\big(W_{2}(x^2)+ V_2(x^2) W_{3}(u_1) \big) p_2 p_{u_2}, \\
			H_2 &= \frac{1}{2} p_1^2 + V_1(x^1) p_1 p_{u_1} +\big(W_{1}(x^1)+ V_1(x^1) W_{3}(u_1) \big) p_1 p_{u_2}, \\
			H_3 &= \frac{1}{2} (p_{u_1}+W_{3}(u_1)p_{u_2})^2, \\
			H_4 &= \frac{1}{2} p_{u_2}^2.
		\end{split}
	\end{equation}
This system generically describes geodesic motion in a 4D configuration space.

\begin{remark}
	More generally, the proposed iterative construction yields higher-dimensional integrable Hamiltonian systems that a priori can serve as explicit toy models for field theory and for the description of anisotropic media.
\end{remark}

\section{St\"ackel lift, Lorentzian geometries and gravitational waves} \label{sec:6}

In Section~\ref{sec:3.1}, we reviewed the Eisenhart lift, which embeds mechanical systems into Bargmann-Eisenhart waves—spacetimes with parallel null vector fields. Here we show how the Stäckel lift naturally extends this construction to Platonic waves \cite[Def.~4.11]{BM2013PRD}, a special class of Kundt spacetimes with potential applications in modified gravity theories.

A Platonic wave is defined as a spacetime admitting a null Killing vector field $\xi$ that becomes parallel with respect to a conformally related metric; the conformal factor is then automatically Lie-constant along $\xi$.

Equivalently, Platonic waves are conformally equivalent to Bargmann-Eisenhart waves while preserving the null Killing symmetry.

The Stäckel approach is systematic: it treats both the conformal factor and the potentials as components of a unified lifting matrix, directly revealing separability. To illustrate, consider the lifting matrix
\begin{equation}\label{Platonic lift}
	\mathcal{L}^{(4)}(\bs{q},\bs{p}):=
	\begin{bmatrix}
		0 & 1/2 & -V_1(x^1) & -W_{1}(x^1) p_1 \\
		1 & -1/2 & -V_2(x^2)& -W_{2}(x^2) p_2 \\
		0 & 0 & 1 & - p_{u} \\
		0 & 0 & 0 & 1 \\
	\end{bmatrix}.
\end{equation}

	From the choice $f_1= \frac{1}{2}p_1^2$, $f_2=\frac{1}{2}p_2^2$, $f_3= \frac{1}{2}p_u^2$, $f_4= p_t$, we obtain (after some algebraic manipulations) the integrable Hamiltonian model
\begin{equation}
	\begin{split}
		H_1 &= \frac{1}{2}\big(p_1^2+p_2^2\big) +\big(V_1(x^1)+V_2(x^2) \big) \left( \frac{1}{2}p_u^2 + p_t p_u\right) + W_1(x^1) p_1 p_t + W_2(x^2) p_2 p_t, \\
		H_2 &= \frac{1}{2} \big(p_1^2 + V_1(x^1) p_{u}^2\big) + V_1(x^1) p_t p_u + W_1(x^1) p_1 p_t, \\
		H_3 &= \frac{1}{2} p_{u}^2+ p_t p_u, \\
		H_4 &= \frac{1}{2} p_t^2.
	\end{split}
\end{equation}
The Stäckel construction ensures that these integrals $\{H_1, H_2, H_3, H_4\}$ are in involution; they are also associated with a symplectic-Haantjes manifold, guaranteeing complete separability.
Setting $\Omega(x^1,x^2) := V_1(x^1)+V_2(x^2)$, this can be rewritten as
\begin{equation}
	\begin{split}
		H_1 &= \frac{1}{2}g^{\mu\nu}p_\mu p_\nu=\Omega(x^1,x^2) \Bigg[\frac{p_1^2+p_2^2}{2\Omega(x^1,x^2)} +\frac{1}{2} p_u^2 +p_t p_u \\
		&\quad + \frac{W_1(x^1)}{\Omega(x^1,x^2)} p_1 p_t + \frac{W_2(x^2)}{\Omega(x^1,x^2)} p_2 p_t \Bigg],
	\end{split}
\end{equation}
revealing the underlying conformal structure. The corresponding spacetime metric, obtained by inverting the $g^{\mu\nu}$ of $H_1$, takes the form
\begin{equation}
	ds^2 = (dx^1)^2 + (dx^2)^2 + \frac{1}{\Delta}\left[ 2\,dW(dt - du) - (dW)^2 + \frac{\|\boldsymbol{W}\|^2}{\Omega}\,du^2 + 2\,du\,dt - dt^2 \right],
\end{equation}
where $dW := W_1(x^1)\,dx^1 + W_2(x^2)\,dx^2$, $\|\boldsymbol{W}\|^2 := W_1^2 + W_2^2$, $\Omega := V_1(x^1) + V_2(x^2)$, and $\Delta := \Omega + \|\boldsymbol{W}\|^2$. Both $\partial_u$ and $\partial_t$ are Killing vector fields, since the metric depends only on $(x^1,x^2)$. Assume $\Omega >0$. The component $g_{uu} = \|\boldsymbol{W}\|^2/(\Omega\Delta)$ vanishes if and only if $W_1 = W_2 = 0$; only then is $\partial_u$ null and the metric conformally a Bargmann–Eisenhart wave, with the conformal structure exhibited by the factorization above. In this case it reduces to $ds^2 = (dx^1)^2 + (dx^2)^2 + \Omega^{-1}(2\,dt\,du - dt^2)$, and when moreover $\Omega \equiv 1$ (i.e., $V_1 + V_2 = 1$) one recovers the Bargmann-Eisenhart case of Section~\ref{sec:3}, where $\partial_u$ is parallel rather than merely Killing. Consequently, $g = \Omega^{-1}\tilde{g}$, where $\tilde{g} = \Omega\left[(dx^1)^2+(dx^2)^2\right] + 2\,dt\,du - dt^2$ is a Bargmann-Eisenhart wave with parallel null vector $\partial_u$; the conformal factor in the sense of \cite[Def.~4.11]{BM2013PRD} is thus $\Omega^{-1}$. For $\|\boldsymbol{W}\|^2 \neq 0$,  instead, $\partial_u$ is spacelike and the geometry is a stationary spacetime.

\subsection{Example: Integrable Platonic wave with flat wave fronts}
To illustrate the physical content and the power of the Stäckel framework, consider the specific case
\begin{equation}
	V_1(x^1) = 0, \quad V_2(x^2) = (x^2)^{-2}, \quad W_1(x^1) = W_2(x^2) = 0.
\end{equation}

This yields a Coriolis-free Platonic wave with the metric:
\begin{equation}
	ds^2 = (dx^1)^2 + (dx^2)^2 + 2(x^2)^2\,dt\,du - (x^2)^2\,dt^2.
\end{equation}

Setting $x^2 = z$ and $x^1 = y$, the wave fronts (the surfaces $u,t=\text{const}$) are flat planes with metric $d\ell^2 = (dy)^2 + (dz)^2$. The spacetime itself is curved: its Ricci scalar $\mathcal{R} = -2/(x^2)^2$ is non-constant, so this is a curved, completely integrable Kundt-class Platonic wave, although it is not an Einstein manifold.

The Hamiltonian system reads
\begin{equation}
	H_1 = \frac{1}{2}(p_1^2 + p_2^2) + \frac{1}{(x^2)^2}\left(\frac{1}{2}p_u^2 + p_t p_u\right),
\end{equation}
with integrals of motion
\begin{equation}
	H_2 = \frac{1}{2}p_1^2, \quad H_3 = \frac{1}{2}p_u^2 + p_t p_u, \quad H_4 = \frac{1}{2}p_t^2.
\end{equation}

In the context of gravitational theories, the key features of our approach are the following. 

\textit{(i) Complete integrability and separability:} The Stäckel construction guarantees that the Hamilton-Jacobi equation separates in the coordinates $(x^1, x^2, u, t)$. This is a rare property for spacetimes with nontrivial curvature.

\textit{(ii) Projection to nonrelativistic dynamics:} Geodesics with $p_u = m \neq 0$ and $p_t = E$ project, with respect to the affine parameter, onto the motion of a particle on the flat plane $\mathbb{R}^2$ subject to the inverse-square potential $(\tfrac12 m^2 + Em)\,(x^2)^{-2}$. The complete separability in the ambient spacetime directly implies separability of the reduced nonrelativistic system.

\textit{(iii) Geometric structure:} As a Platonic wave, this spacetime belongs to the Kundt class—it admits a geodesic, expansion-free, shear-free, and twist-free null congruence. Although the wave fronts are flat, the null Killing field $\partial_u$ is not parallel, unlike in the case of $pp$-waves.

\textit{(iv) Extensions to modified gravity contexts:} Platonic-wave geometries can arise as solutions in modified gravity theories (e.g., $f(R)$ gravity with matter coupling). Whether our specific examples satisfy such field equations is left for future work.

\begin{remark}
	These completely integrable Platonic-wave geometries illustrate how the Stäckel construction naturally generates interesting Kundt-class solutions; physical realizability in modified gravity remains an open problem, as for other Kundt-class geometries such as AdS gyratons \cite{Frolov2005}.
\end{remark}

\section{Non-geodesic flows and higher-order integrable systems} \label{sec:7}
	
	The St\"ackel lift procedure introduced earlier is not limited to producing geodesic lifts. We show the versatility of the method by discussing three of the infinitely many possible non-geodesic flows that can be constructed.
	
	\subsection{A quadratic and non-homogeneous system}
	First, observe that a two-dimensional natural system can also be lifted by means of the matrix \eqref{eq:4.1} to a non-geodesic, still quadratic integrable Hamiltonian system by choosing as St\"ackel functions $f_1= \frac{1}{2}p_1^2 +U_1(x^1)$, $f_2=\frac{1}{2}p^{2}_{2}+U_2(x^2)$, $f_3= p_u+U_3(u)$. In this way, we define the system
	\begin{equation}
		\begin{split}
			H_1 &= \frac{1}{2}(p_1^2+p_2^2) + V_1(x^1) p_1 p_u+ V_2(x^2) p_2 p_u+ (V_1(x^1)p_1+V_2(x^2)p_2) U_3 (u) \\  &+ U_1(x^1)+U_2(x^2), \\
			H_2 &= p_1^2+2 V_1(x^1) p_1 p_u+ 2 V_1(x^1) U_3(u) p_1+ 2 U_1(x^1), \\
			H_3 &= (p_u+ U_3(u))^2.
		\end{split}
	\end{equation}
	This system can be made homogeneous by applying a standard Eisenhart lift to it.
	\subsection{Higher-order or transcendental systems}
	
	Due to the arbitrariness of the functions in the lifted matrix \eqref{eq:LSM}, higher-order integrable and separable systems admitting non-trivial, even transcendental integrals in the momenta can be easily generated.
	
	As a simple realization of this idea, from the simple deformation of matrix \eqref{eq:4.1} given by 
	\begin{equation} \label{eq:6.2}
		\mathcal{L}^{(3)}(\bs{q},\bs{p})=
		\begin{bmatrix}
			0 & 1/2 & -V_1(x^1) e^{p_1} \\ 1 & -1/2 & -V_2(x^2)e^{-p_2} \\ 0 & 0 & 1
		\end{bmatrix},
	\end{equation}
	and by choosing $f_1=\frac{1}{2}p^{2}_{1}$, $f_2=\frac{1}{2}p^{2}_{2}$, $f_3= p_u$, 
	we obtain (after some manipulations) the integrable model 
	\begin{equation}
		\begin{split}
			H_1 &= \frac{1}{2}(p_1^2+p_2^2) + V_1(x^1) e^{p_1} p_u+ V_2(x^2) e^{-p_2} p_u, \\ 
			H_2 &= \frac{1}{2} p_1^2 + V_1(x^1) e^{p_1} p_u, \\
			H_3 &= \frac{1}{2} p_u^2.
		\end{split}
	\end{equation}

	\subsection{Lifted Hamiltonians with linear terms}
	
	The St\"ackel lift can also be easily accommodated to produce lifted systems admitting linear terms in the momenta. Systems of this kind can be associated with gyroscopic or, in suitable time-dependent or contact formulations, dissipative dynamics; they are physically interesting on their own.
	A straightforward way is to consider a lifting matrix of the form

	\begin{equation}
		\mathcal{L}^{(4)}(\bs{q},\bs{p}):=
		\begin{bmatrix}
			0 & 1/2 & -V_1(x^1)+ p_1 & -W_{1}(x^1)+ p_1 \\ 1 & -1/2 & -V_2(x^2)+ p_2 & -W_{2}(x^2)+ p_2 \\ 0 & 0 & 1 & -W_{3}(u_1) \\ 0 & 0 & 0 & 1 \\
		\end{bmatrix},
	\end{equation}
	where $W_{1}(x^1)$, $W_{2}(x^2)$, $W_{3}(u_1)$ are arbitrary functions. 
	By choosing as St\"ackel functions $f_1= \frac{1}{2}p_1^2$, $f_2=\frac{1}{2}p_2^2$, $f_3= p_{u_1}$, $f_4= p_{u_2}$, we obtain a different quadratic integrable model with linear terms in 4D (for brevity, we omit the explicit expressions of the Hamiltonians).

\begin{remark}
	The physical motivation for the analysis of this family of models relies on the fact that several important classical systems admit higher-order integrals of motion, such as the three-particle Calogero model \cite{TT2016SIGMA}, the Post-Winternitz system \cite{TT2022AMPA}, etc. In addition, they correspond to higher-order quantum operators, making explicit classical examples valuable for the study of exactly solvable quantum models.
\end{remark}

\section{St\"ackel lift for cylindrical separable magnetic systems} \label{sec:8}
The separation of natural Hamiltonians with scalar and vector potentials on 
Riemannian manifolds was systematically studied by Benenti, Chanu, and Rastelli 
\cite{Benenti2001} using the Riemannian Eisenhart lift discussed in Section 
\ref{sec:3}. They provided a complete geometric characterization in terms of 
characteristic Killing pairs $(D, K)$ on the configuration space, proving that 
geodesic separation in the extended space corresponds to gauge separation 
 of the original Hamiltonian.

In this section, we apply our Stäckel lift to the three classes of magnetic 
systems separable in cylindrical coordinates classified in \cite{Benenti2001} 
(see also \cite{Shapovalov} for quantum systems). These systems were rederived 
from the viewpoint of Haantjes geometry in \cite{KRTT2024PRA}, but 
Stäckel matrices were explicitly provided only for Case 3. Here we complete the picture by 
constructing the Stäckel structures for all three cases, thereby filling a gap 
from our previous work.

Our analysis reveals that all three cases admit Stäckel formulations, but with important differences. Case 3 allows the simple treatment from \cite{KRTT2024PRA}, where adding $H_{3}^2$  preserves dynamics on level sets. Cases 1 and 2, however, require incorporating the coefficient $\frac{1}{r^2}$ into the change of basis, which modifies the Hamiltonian itself rather than merely reparametrizing the dynamics. We thus demonstrate how to embed classical separable magnetic systems within the St\"ackel framework. The required transformations reflect the geometric structure of the underlying coordinate system.

The St\"ackel matrices corresponding to the three separable cases read\footnote{The term $V_1(r)/r^2$ in $S_2$ is necessary to obtain $p_\phi/r^2$ in the Hamiltonian as in \cite{KRTT2024PRA}.}

\begin{equation} \label{S3 cyl}
	\begin{split}
		S_1 &= \left[\begin{array}{ccc}
			1 & -\frac{f(r)}{r^{2}} & -V_1(r) \\
			0 & 1 & 0 \\
			0 & 0 & 1 
		\end{array}\right], \quad
		S_2 = \left[\begin{array}{ccc}
			1 & -1 & -\frac{V_1(r)}{r^2} \\
			0 & 0 & 1 \\
			0 & 1 & -V_3(z) 
		\end{array}\right], \\
		S_3 &= \left[\begin{array}{ccc}
			1 & -\frac{1}{r^{2}} & -V_1(r) \\
			0 & 1 & -V_2(\phi) \\
			0 & 0 & 1 
		\end{array}\right].
	\end{split}
\end{equation}

For each case discussed below, we begin with the general lift\footnote{We choose the minus sign of $S_{j4}$ for convenience.} 

\begin{equation}
	\mathcal{L}^{(4)}_i = \left[\begin{array}{ccc|c}
		& & &-S_{14}(r, p_r) \\
		& S_i & & -S_{24}(\phi, p_\phi) \\
		& & &- S_{34}(z, p_z) \\ \hline
		0 & 0 & 0 & 1
	\end{array}\right],
\end{equation}
and examine admissible choices that yield quadratic Hamiltonians, thereby remaining within the conventional separation on the configuration space. 

\vspace{2mm}

\begin{remark}
	The systems that we will rederive in the following analysis describe the dynamics of charged particles in cylindrically separable magnetic fields, namely the canonical setting for magnetic traps and confinement geometries. More general constructions, which lead to momentum-dependent metrics and are likewise of physical interest (see Remark \ref{rmk momentum dep}), will be explored in future work.
\end{remark}

\textit{Case 1.} The matrix $S_1$ corresponds to the case with two cyclic variables $(\phi, z)$, that is, two linear integrals of motion. Accordingly, two of the original St\"ackel functions are linear in the momenta, namely $f_1 = p_r^2 + U_r(r)$, $f_2 = p_{\phi}$, $f_3 = p_z$, along with $f_4 = p_u$.
Using this ansatz in \eqref{eq:LSH}, we obtain the separable system
\begin{equation}
	\begin{split}
		H_1 &= p_r^{2}+U_r(r)+\frac{f(r) p_\phi}{r^{2}}+V_1(r) p_z \\
		&\quad + \left(V_1(r) S_{34}(z, p_z) +S_{14}(r, p_r)+\frac{f(r) S_{24}(\phi, p_\phi)}{r^{2}}\right) p_u, \\
		H_2 &= p_\phi +S_{24}(\phi, p_\phi) p_u, \\
		H_3 &= p_z +S_{34}(z, p_z) p_u, \\
		H_4 &= p_u.
	\end{split}
\end{equation}

To ensure that the integrals associated with cyclic coordinates remain linear in momenta, we impose the conditions $S_{24} = S_{24}(\phi)$ and $S_{34} = S_{34}(z)$, independent of the momenta. These functions can then be eliminated by means of the canonical transformation
\begin{equation}
	(p_{\phi}\to p_{\phi}+S_{24}(\phi)p_u, \quad p_z\to p_z+S_{34}(z) p_u, \quad u\to u-\int S_{34}(z) \mathrm{d} z-\int S_{24}(\phi) \mathrm{d} \phi),
\end{equation}
which is the \textit{Riemannian lifted gauge transformation}. Therefore, we set both $S_{24}$ and $S_{34}$ to zero. 

The linear momentum terms in $H_1$ correspond to the magnetic field and cannot be removed by any choice of the lift functions $S_{j4}$, as they are coupled to $p_u$. Consequently, the St\"ackel lift does not produce an Eisenhart-lifted system.

To preserve the quadratic structure of $H_1$, we set $S_{14}=\sigma_1(r)+\sigma_2(r) p_r$. Then, $H_1$ can be interpreted as a magnetic Hamiltonian in which certain quadratic diagonal terms in the metric are absent. To render the Hamiltonian complete, we adopt a modified Stäckel basis:
\begin{equation}
	(\tilde{H}_1:=H_1+\frac{1}{r^2}H_2^2+H_3^2,\tilde{H}_2:=H_2,\tilde{H}_3:=H_3,\tilde{H}_4:=H_4).
\end{equation}

This modification incorporates the coefficient $\frac{1}{r^2}$ from the cylindrical coordinate system to complete the quadratic form. However, unlike the simple addition of $p_z^2$ in Case 3 (which preserves dynamics on level sets of the integrals), here we are adding $\frac{c_2^2}{r^2} + c_3^2$ on the level sets $H_2 = c_2$, $H_3 = c_3$. Since $r$ is a dynamical variable, the term $\frac{c_2^2}{r^2}$ varies along trajectories, meaning $\tilde{H}_1$ and $H_1$ define different energy surfaces and therefore describe genuinely different dynamical systems.

Nevertheless, this construction demonstrates that Case 1 can be embedded into Stäckel theory, albeit at the cost of modifying the Hamiltonian. The modified $\tilde{H}_1$ lives naturally on the warped Riemannian manifold $Q_1[r,z]\times_\lambda Q_2[\phi]$ with metric structure $g=g_1(r,z)+r^2 g_2(\phi)$, where $\lambda(r,z)=r^2$ is the warp factor, and admits separation in the sense of \cite{Benenti2001}.

The term independent of momenta, $\sigma_1(r)$, represents a new contribution to the scalar potential of the original system. It depends solely on the non-ignorable coordinate $r$ and acts as a Stäckel multiplier, i.e., it has the same form as a separable scalar potential, as required by the standard theory of separation for magnetic Hamiltonians \cite{Benenti2001}. This term cannot be removed by any canonical transformation; such transformations can only result in a scalar potential term that yields the same equations of motion.

The momentum-dependent term $\sigma_2(r) p_r p_u$ introduces a non-diagonal metric component $g^{ru} = \frac{1}{2}\sigma_2(r)$. Although we can make the metric diagonal via the transformation $p_r \mapsto p_r + \frac{1}{2}\sigma_2(r) p_u$, this term cannot be entirely removed, when $\sigma_2$ is non-constant, reflecting the non-vanishing curvature of the lifted space.

\vspace{2mm}

\textit{Case 2.} Only $\phi$ is cyclic; thus, $p_\phi$ is a linear integral of motion. The corresponding St\"ackel functions are $f_1=p_r^2$, $f_2=p_{\phi}$, $f_3=p^2_z$, and the new one $f_4=p_u$. This implies the existence of a single component of the vector potential in the original system, depending on two variables: $A_\phi=-\frac{1}{2}(V_3(z) r^2 + V_1(r))$.

We use \eqref{eq:LSH} to obtain our separable Hamiltonians\footnote{Because $S_2$ is not upper triangular, equation \eqref{AKN} changes the order of the Hamiltonians $H_2, H_3$.}
\begin{equation}
	\begin{split}
		H_1 &= p_r^2 + \frac{(V_3(z) r^2 + V_1(r)) p_\phi}{r^2} + p_z^2 \\
		&\quad + \frac{[(V_3(z) r^2 + V_1(r)) S_{24}(\phi, p_\phi) + r^2( S_{14}(r, p_r)+S_{34}(z, p_z))] p_u}{r^2}, \\
		H_2 &= p_\phi + S_{24}(\phi, p_\phi) p_u, \\
		H_3 &= p_z^2 +V_3(z) p_\phi + (V_3(z) S_{24}(\phi, p_\phi) + S_{34}(z, p_z)) p_u, \\
		H_4 &= p_u.
	\end{split}
\end{equation}

Following analogous reasoning, we set $S_{34} = S_{34}(z)$ and remove it via a gauge transformation. To preserve the quadratic structure of $H_1$, we choose $S_{14}(r, p_r) = \sigma_1(r) + \sigma_2(r) p_r$ and $S_{24}(\phi, p_\phi) = \tau_1(\phi) + \tau_2(\phi) p_\phi$, which yields the scalar potential term of the original system, $U = \frac{\sigma_1(r) r^2 + (V_3(z)r^2 + V_1(r))\tau_1(\phi)}{r^2}$ and the non-diagonal metric terms $g^{ru} = \frac{1}{2}\sigma_2(r)$, $g^{\phi u} = \frac{1}{2}\frac{V_3(z)r^2 + V_1(r)}{r^2}\tau_2(\phi)$.

The change of Stäckel basis that completes the quadratic form,  provided $\tau_2(\phi)=0$ (for $\tau_2(\phi)\neq 0$, $\tilde{H}_1$ is quartic), is $(\tilde{H}_1:=H_1+\frac{1}{r^2}H_2^2,\tilde{H}_2:=H_2,\tilde{H}_3:=H_3,\tilde{H}_4:=H_4).$ As in Case 1, this modification alters the Hamiltonian rather than merely reparametrizing the original dynamics.

\vspace{2mm}

\textit{Case 3.} Only $z$ is cyclic, hence $p_z$ is the sole  integral linear in the momenta of the original 3D system. The associated St\"ackel functions are $f_1=p_r^2,$ $f_2=p_{\phi}^2$, $f_3=p_z$ and, as always, $f_4=p_u$. Here we lift matrix $S_3$ from \eqref{S3 cyl}.

The resulting Hamiltonians are
\begin{equation}
	\begin{split}
		H_1 &= p_r^2 + \frac{p_\phi^2}{r^2} + \frac{(V_1(r) r^2 + V_2(\phi)) p_z}{r^2} \\
		&\quad + \frac{[(V_1(r) r^2+ V_2(\phi)) S_{34}(z, p_z) + S_{14}(r, p_r) r^2 + S_{24}(\phi, p_\phi)] p_u}{r^2}, \\
		H_2 &= p_\phi^2 + V_2(\phi) p_z + (V_2(\phi) S_{34}(z, p_z) + S_{24}(\phi, p_\phi)) p_u, \\
		H_3 &= p_z + S_{34}(z, p_z) p_u, \\
		H_4 &= p_u.
	\end{split}
\end{equation}
By analogy, we set $S_{34} = S_{34}(z)$ and remove it via a gauge transformation. Keeping the quadratic structure of $H_1$ requires $S_{14}(r, p_r) = \sigma_1(r) + \sigma_2(r) p_r$ and $S_{24}(\phi, p_\phi) = \tau_1(\phi) + \tau_2(\phi) p_\phi$, yielding the scalar potential term $U = (\sigma_1(r) r^2 + \tau_1(\phi))/r^2$ and the non-diagonal metric terms $g^{ru} = \frac{1}{2}\sigma_2(r)$, $g^{\phi u} = \frac{1}{2}\frac{\tau_2(\phi)}{r^2}$.

The change of St\"ackel basis that incorporates the missing $p_z^2$ term is a simple linear combination: $(\tilde{H}_1:=H_1+H_3^2,\tilde{H}_2:=H_2,\tilde{H}_3:=H_3,\tilde{H}_4:=H_4).$ This change of basis, which leads to the St\"ackel formulation of the original 3D system, was found in \cite{KRTT2024PRA}.

\begin{remark}
	The analysis above reveals a subtle hierarchy among separable magnetic systems in cylindrical coordinates. Case 3, where only $z$ is cyclic, admits a straightforward Stäckel formulation with the simple modification $\tilde{H}_1 = H_1 + H_3^2$, which preserves the dynamics on level sets of $H_3 = p_z$. In contrast, Cases 1 and 2 require incorporating the  coefficient $\frac{1}{r^2}$ into the change of basis, which modifies the Hamiltonian and therefore the dynamics.

The Haantjes geometric approach developed in \cite{KRTT2024PRA} handles all three cases uniformly without requiring such modifications, demonstrating the advantage of coordinate-free geometric characterizations of separability for systems that resist traditional separation methods.
	
	More generally, the separability theory in \cite{Benenti2001}, which applies to all (pseudo)Riemannian manifolds without second-class null coordinates, implies that—in adapted coordinates—the linear terms in the momenta are always associated with a Killing vector and its corresponding first-order integral. The Lamé-type coefficients used in our changes of Stäckel basis depend only on the essential coordinates and therefore commute with the conserved integrals that are added, which is why these modifications yield separable systems (albeit with inequivalent Hamiltonians).
\end{remark}

\section{Discussion and future perspectives} \label{sec:9}

The main contribution of this work is to establish the Stäckel lift as a unified geometric framework for constructing integrable Hamiltonian systems of dimension $n+k$ from lower-dimensional seed systems. Theorem~\ref{theo:1} demonstrates that the Riemannian Eisenhart lift emerges as a special case of this construction, revealing a deeper structural principle underlying these long-studied geometries. The Stäckel perspective reframes the Eisenhart lift not as an ad hoc embedding procedure but as an instance of generalized Stäckel geometry, where the lifting matrix plays the role of extended Stäckel structure. This unification suggests that the special role of the Eisenhart lift in embedding mechanical systems as geodesic flows is a consequence of a more general principle: any St\"ackel-separable system on an $n$-dimensional manifold can be lifted to an $(n+k)$-dimensional integrable system by extending its Stäckel matrix.

Theorem~\ref{maintheorem} establishes that Stäckel-lifted systems naturally carry non-trivial symplectic-Haantjes structure, while the Haantjes operators can be reconstructed from those of the seed system. This result deepens the connection between separation of variables and Haantjes geometry initiated in recent work on $\omega\mathscr{H}$ manifolds \cite{TT2021JGP,TT2022AMPA,RTT2022CNS,RTT2024AMPA}. Crucially, momentum-dependent Stäckel lifts produce non-projectable Haantjes operators — in particular, they cannot be obtained by lifting Killing tensors from configuration space (for a related discussion, see \cite{BKM2025stackel,BKM2026duality}). This non-projectability is an interesting phenomenon that reflects the new geometric structure of momentum-dependent lifting matrices.

The existence of such structures motivates a further classification and understanding of when and how Haantjes geometry is not related to Killing geometry, and is therefore more general. This question has implications for the structure of integrable systems on cotangent bundles.

The applications we presented demonstrate the versatility of the Stäckel lifting framework. Momentum-dependent lifting matrices can produce Platonic-wave geometries \cite{BM2013PRD}---Kundt-class spacetimes with nontrivial conformal structure and complete integrability of the geodesic flow. Although our examples do not satisfy the vacuum Einstein equations, they point to a systematic method for constructing integrable spacetime models with Kundt geometry, which may be solutions in modified gravity theories. The construction of Hamilton and generalized Hamilton spaces via momentum-dependent metrics connects to an active area of research in deformed relativity and quantum gravity phenomenology \cite{Miron2001,Voicu2021,RSS2024,BBGLP2015PRD,CLR2022PRD,Albuquerque}. The recovery of separable magnetic systems in cylindrical coordinates within the Stäckel framework \cite{KRTT2024PRA} demonstrates that the lifting procedure accommodates both geodesic and non-geodesic Hamiltonians, expanding the method's scope.

Several natural questions remain open.

We wish to ascertain whether the Platonic-wave examples constructed here satisfy the field equations of any modified gravity theory (e.g., $f(R)$ gravity or higher-derivative theories), and if so, what physical interpretation they admit.  

It would also be interesting to understand if the non-projectable Haantjes structures arising from momentum-dependent lifts can be fully classified, according to their algebraic properties. 

The non-geodesic Hamiltonians described in \S\ref{sec:7} can possess transcendental integrals of motion; a natural direction for future research is to investigate whether these models admit quantum analogues in the form of non-trivial symmetry operators, and whether such operators generate new exactly or partially solvable  quantum models. 

Finally, the iterative lifting procedure suggests the existence of a "ladder" of integrable systems of increasing dimension. Understanding the geometric structure of this ladder, and the constraints it imposes, may yield further insight into the structure of integrable Hamiltonian systems.

\section {Acknowledgement}	
	

We thank Professors Giuseppe Marmo and Giorgio Tondo for a reading of the manuscript and useful discussions. 
	
	O.K.'s postdoctoral fellowship is financed by the project ``Fostering ICMAT's Strategic Scientific Lines'' (202450E223). 
	
The research of P.T. has been supported by the Project PID2024-156610NB-I00 ``Haantjes geometry and integrable systems'' of Ministerio de Ciencia, Innovaci\'on y Universidades and by the Severo Ochoa Programme for Centres of Excellence in R\&D
(CEX-2023-001347-S), Ministerio de Ciencia, Innovaci\'{o}n y Universidades and Agencia Estatal de Investigaci\'on, Spain. 
	
P.T. is a member of the Gruppo Nazionale di Fisica Matematica (GNFM) of the Istituto Nazionale di Alta Matematica (INdAM).

AI-assisted technology (Claude from Anthropic and  ChatGPT from OpenAI) was used to improve readability and language of the manuscript and to assist in validating results. All mathematical derivations, proofs, results, and scientific conclusions are solely the responsibility of the authors.

\end{document}